\documentclass[]{interact}

\usepackage{epstopdf}
\usepackage[caption=false]{subfig}
\usepackage{amsmath,amssymb,amsthm}           
\usepackage{xcolor}

\usepackage[numbers,sort&compress,merge]{natbib}
\bibpunct[, ]{[}{]}{,}{n}{,}{,}

\theoremstyle{plain}

\newtheorem{thm}{Theorem}[section]

\newtheorem{ass}{Assumption}
\newtheorem{rem}[thm]{Remark}
\newtheorem{defn}[thm]{Definition}

\newcommand{\R}{{\mathbb R}}

\newcommand{\cL}{{\mathcal L}}

\newcommand{\cO}{{\mathcal O}}

\newcommand{\cQ}{{\mathcal Q}}

\newcommand{\cN}{{\mathcal N}}

\newcommand{\cA}{{\mathcal A}}
\newcommand{\cX}{{\mathcal X}}

\newcommand{\bE}{{\mathbf E}}
\newcommand{\bP}{{\mathbb P}}
\newcommand{\bQ}{{\mathbb Q}}

\newcommand{\KL}{{\it KL}}

\newcommand{\wrt}{with respect to }

\begin{document}

\title{Computationally feasible bounds for the free energy of nonequilibrium steady states, applied to simple models of heat conduction}

\author{Luigi Delle Site\textsuperscript{a}, Carsten Hartmann\textsuperscript{b}\thanks{Corresponding author: carsten.hartmann@b-tu.de}\\
\affil{\textsuperscript{a}Freie Universität Berlin, Institut für Mathematik, 14195 Berlin, Germany\\ \textsuperscript{b}Brandenburgische Technische Universität Cottbus-Senftenberg\\ Institut für Mathematik, 03046 Cottbus, Germany}}%
\thanks{This work was partly supported by the DFG Collaborative Research Center 1114 ``Scaling Cascades in Complex Systems'', project No.235221301, Projects A05 (C.H.) ``Probing scales in equilibrated systems by optimal nonequilibrium forcing'' and C01 (L.D.S.) ``Adaptive coupling of scales in molecular dynamics and beyond to fluid dynamics''. 
		One author (C.H.) acknowledges support by the German Federal Government, the Federal Ministry of Education and Research and the State of Brandenburg within the framework of the joint project ``EIZ: Energy Innovation Center'' (project numbers 85056897 and 03SF0693A) with funds from the Structural Development Act (Strukturstärkungsgesetz) for coal-mining regions.}

\date{\today}

\maketitle
\begin{abstract}
In this paper we study computationally feasible bounds for relative free energies between two many-particle systems. Specifically, we consider systems out of equilibrium that do not necessarily satisfy a fluctuation-dissipation relation, but that nevertheless admit a nonequilibrium steady state that is reached asymptotically in the long-time limit. The bounds that we suggest are based on the well-known Bogoliubov inequality and variants of Gibbs' and Donsker-Varadhan variational principles.   
As a general paradigm, we consider systems of oscillators coupled to heat baths at different temperatures. For such systems, we define the free energy of the system relative to any given reference system (that may or may not be in thermal equilibrium) in terms of the Kullback-Leibler divergence between steady states. 
By employing a two-sided Bogoliubov inequality and a mean-variance approximation of the free energy (or cumulant generating function, we can efficiently estimate the free energy cost needed in passing from the reference system to the system out of equilibrium (characterised by a temperature gradient). A specific test case to validate our bounds are harmonic oscillator chains with ends that are coupled to Langevin thermostats at different temperatures; such a system is simple enough to allow for analytic calculations and general enough to be used as a prototype to estimate, e.g., heat fluxes or interface effects in a larger class of nonequilibrium particle systems.
 \end{abstract}

\section{Introduction}
Particle systems subject to thermal gradients are prime examples of nonequilibrium systems that can reach a nonequilibrium steady state in the long-time limit. Such systems are of ongoing interest in molecular nanoscience as they give rise to a variety of nonequilibrium phenomena that can be used for developing future technology.
For example, experimental studies have detected non-uniform
temperature fields around membranes of living cells \cite{TTT1, TTT1b,TTT1c}. One key question in this context concerns how temperature gradients affect the geometry of membranes. It is thought that even small temperature variations could generate interesting phenomena, such as unexpected shape responses. These results suggest that the shape of a membrane can be tuned externally by controlling the temperature field in its immediate vicinity giving rise to the possibility of a wide range of bio-technological applications \cite{TTT2}. Another explicit example of high interest in current research concerns ionic liquids. 
Ionic liquids are molecular liquids composed
of cations (positively charged molecules) and anions (negatively charged molecules); their relevance in current technology is growing steadily due to their many interesting properties. Most importantly, they can be used as green solvents as they are easily recyclable \cite{TTT3}. One possibility
for the recycling process is to induce a response via thermal gradients so that a phase separation (ionic liquid-rich phase) occurs at some critical level of nonequilibrium forcing, thus separating the ionic-liquid from other components, e.g., water \cite{TTT4}. Besides these two specific example, there exists a wide range of phenomena of interest in this field, from the evaporation of liquids \cite{TTT5a,TTT5b}, to the thermomolecular orientation of nonpolar fluids\cite{TTT6}, effects of thermo-phoresis \cite{TTT7a,TTT7b} and thermo-osmosis \cite{TTT8}, separation in liquid mixtures \cite{TTT9a,TTT9b}, diffusion of polymers in a solvent \cite{TTT10a,TTT10b}, heat transfer in protein–water interfaces \cite{TTT11}, or the polarization of water \cite{TTT12} to name just a few. We should mention that particle systems coupled to different heat baths are classical stochastic optimisation tools that nowadays play a major role in machine learning where such models are used for training of high-dimensional neural networks; see, e.g.~\cite{marinari1992simulated,Leimkuhler2019,breiten2021stochastic}

There is a variety of numerical methods from molecular simulation that can deal with such nonequilibrium processes, treating the heat bath directly \cite{giovanni-noneq,noneq1,noneq2,prlabbas,TTT13,TTT14,TTT15} or in a multiscale fashion \cite{TTT16,TTT17}. Nonetheless, a key difficulty in many applications is the calculation of thermodynamic averages or time-dependent quantities, such as heat fluxes, large deviations rate functions, or entropy production rates, since the simulations often suffer from poor convergence. In this situation, enhanced sampling methods or variance minimisation techniques come into play, but theses methods are often tied to a thermodynamic equilibrium setting; see \cite{henin2017} and the references therein.   

\subsection{Own contribution and outline of this article}

In this work, we study---in a restricted setting--- estimators for free energies, cumulant generating functions, or large deviations rate functions that are often the quantities of interest in the simulation of nonequilibrium systems. Unlike the thermodynamic free energy that does not have an a priori meaning for nonequilibrium systems, the cumulant generating function and its close relative, the large deviations rate function, can be unambigously defined. Here we exploit that the free energy has the form of a cumulant generating function and suggest computationally cheap estimates or bounds that are based on variations of the same theme: the Gibbs variational principle for the cumulant generating function (or free energy) and its dual formulation, the Donsker-Varadhan variational principle for the Kullback-Leibler divergence (or relative entropy).  An advantage of these varisational formulations and the resulting free nenergy bounds is that they can be easily generalised to the infinite dimensional setting, so as to include, e.g., hydrodynamic limits or problems on path space. 

Yet, for the sake of the argument, we do all this in a rather simplified setting: we consider a universal and simple many-particle system subject to a thermal gradient; specifically, we consider harmonic oscillator chains in contact with two different heat baths, and we even restrict ourselves to the one-dimensional case which is a simplified generic prototype of the examples above. Despite its simplicity, the harmonic model  can be useful for predicting basic trends of the physics associated to the heat flow and thus serves as a reference for building more complex models. 
For multiscale approaches this particular system has been used in a hybrid fashion together with fully resolved molecular dynamics \cite{TTT16,TTT17}. Moreover, it is routinely used as a model for simulating low dimension materials \cite{TTT18} or modern quantum technological devices \cite{TTT19}, to give just two examples. From a theoretical point of view, the study of heat transport in low-dimensional oscillator lattices has generated a massive effort in research as it is one of the key paradigms to study the mechanism of heat transport that is in many regards still not fully understood \cite{TTT20}. 

Specifically, in this paper, we extend recent results on the two-sided Bogoliubov inequality as a computational tool to estimate interface free energies and finite size effects in equilibrium systems \cite{DelleSite2017,Reible2022,Reible2023}. We suggest a formulation that is applicable to nonequilibrium situations  and derive further approximations, such as the mean-variance approximation that is known from mathematical finance \cite{newbery1988accuracy} and a finite-time bound. We illustrate our findings with the simplest possible model of heat conduction: a harmonic chain of oscillators with pinned ends coupled to two different heat baths. 

The article is structured as follows: In Section \ref{sec:osc}, we introduce the basic Langevin model that will be used throughout the paper to streanmline the discussion. In Section \ref{sec:freeEnergy} a nonequilibrium free energy concept is introduced and the corresponding bounds and approximations are derived. The empirical validation of the bounds and the application to simple models of heat conduction is presented in Section \ref{sec:cases}, before the article concludes in Section \ref{sec:conc}. Appendix \ref{sec:fir} records an auxiliary technical result.

\section{Oscillators coupled to two different heat baths}\label{sec:osc}

We consider a homogeneous assembly of $N$ coupled particles in $\R^d$, the dynamics of which is governed by the Hamiltonian
\begin{equation}\label{H}
	H(q,p) = \sum_{i=1}^N  \frac{|p_i|^2}{2 m_i} + V(q_i) + \sum_{j>i} U(q_{j}-q_i)\,.
\end{equation}
Here $q=(q_1,\ldots,q_N)$ and $p=(p_1,\ldots,p_N)$, where $q_i\in\R^d$ is the displacement of the $i$-th particle relative to some (arbitrary) reference positions, $m_i>0$ is its mass, and $p_i\in\R^d$ is its conjugate momentum. We consider only $d=1,2,3$.
The functions $V,U\colon\R^d\to\R$ are smooth pinning and interaction potentials that are assumed to be bounded below, with $V$ polynomially growing at infinity.

We suppose that if the system is in contact with a heat bath at constant inverse temperature $\beta=(k_BT)^{-1}$, with $k_B$ denoting Boltzmann's constant and $T>0$ the temperature, then the dynamics admits a unique invariant measure $\mu_{\rm e}$ with density 
\begin{equation}\label{canDense}
	\rho_{\rm e} = \frac{1}{Z_{\rm e}}e^{-\beta H}\,,\quad Z_{\rm e} = 	\int e^{-\beta H}\,dq dp\,.
\end{equation}
(Here and in the following, we set $k_B=1$.) Existence of an invariant measure on $\R^{2Nd}\times\R^{2Nd}$ requires that $V\neq 0$, for otherwise the Hamiltonian is invariant under translations $q_i\mapsto q_i + u$ by an arbitrary vector $u\in\R^d$, and this implies $Z=\infty$ since a translation-invariant Hamiltonian is no longer coercive. 

In this paper, we assume that the system is in contact with \emph{two} heat baths---more specifically, we let $L,R\subset\{1,\ldots,N\}$  be two disjoint subsets of the particle index set, and consider the Langevin dynamics 
\begin{equation}\label{langNoneq}
	\begin{aligned}
		dq_i & = \frac{\partial H}{\partial p_i}dt\,, \quad i=1\ldots N\\
		dp_i & = -  \frac{\partial H}{\partial q_i}dt\,, \quad i\notin L,R \\
		dp_i & = - \left(\frac{\partial H}{\partial q_i}dt + \gamma \frac{\partial H}{\partial p_i}\right)dt + \sqrt{2\gamma T_L}dW_i \,, \quad i\in L \\
		dp_i & = - \left(\frac{\partial H}{\partial q_i}dt + \gamma \frac{\partial H}{\partial p_i}\right)dt + \sqrt{2\gamma T_R}dW_i \,, \quad i\in R \\
	\end{aligned}
\end{equation}
Here $T_L,T_R>0$ denote the temperatures for the particles in $L$ and $R$. One such example is the one-dimensional linear chain, with $L=\{1\}$ and $R=\{N\}$, with pinned ends  and harmonic nearest neighbour interactions: 
\begin{equation}\label{linearChain}
	H(q,p) = \sum_{i=1}^N  \frac{p_i^2}{2 m_i} + \frac{\alpha}{2}\left(q_1^2+q_N^2\right) + \frac{1}{2}\sum_{i=1}^{N-1} (q_{i+1}-q_i)^2\,,\quad \alpha>0\,.
\end{equation}
For $T_L\neq T_R$ the system admits a unique Gaussian invariant measure that is different from $\rho_{\rm e}$ unless $T_L=T_R=T$. (See below for details.)

\subsubsection*{Compact writing}

We will often use a more compact notation to display the nonequilibrium dynamics (\ref{langNoneq}): setting $x=(q,p)\in\R^{2Nd}$ and introducing the matrices
\begin{equation}
	J = \begin{pmatrix}
		0 & I \\ -I & 0
	\end{pmatrix}\,,\quad D = \begin{pmatrix}
	0 & 0 \\ 0 & \hat{\gamma}
	\end{pmatrix}\,,\quad C = \begin{pmatrix}
	0 \\ \hat{\sigma}
	\end{pmatrix},
\end{equation}
with the obvious definitions of the square matrices $\hat{\gamma},\,\hat{\sigma}\in\R^{(|L|+|R|)d\times (|L|+|R|)d}$, the Langevin equation (\ref{langNoneq}) can be recast as
\begin{equation}\label{langPHS}
	dx(t) = (J-D)\nabla H(x(t)) + CdW(t)
\end{equation}
where $W=(W_L,W_R)$ is an $(|L|+|R|)d$-dimensional Brownian motion that comprises all the noisy components of the nonequilibrium system. 

\subsection{Equilibrium distribution} 

Using the representation (\ref{langPHS}), it is easy to see that $\rho_{\rm e}$ as given by (\ref{canDense}) is an invariant density if $T_L=T_R=T$. To this end, we define the second-order differential operator 
\begin{equation}
	L = \frac{1}{2} CC^\top\colon\nabla^2 + (J-D)\nabla H\cdot\nabla 
\end{equation} 
where $A\colon B={\rm tr}(A^\top B)$ denotes the inner product between matrices (also called: Frobenius inner product). Further letting $\mu(f)$ denote the expectation of a smooth, compactly supported test function $f$ \wrt a probability measure $\mu$, then $\mu=\mu_{\rm e}$ is invariant under the dynamics generated by $L$ if $\mu(Lf)=0$: 
\begin{equation}
	\mu_{\rm e}(Lf) = \int (Lf)\, d\mu_{\rm e} =  \int (Lf) \rho_{\rm e} \,dx =  \int f (L^*\rho_{\rm e}) \,dx = 0\,,
\end{equation} 
where $L^*$ is the formal $L^2$-adjoint of $L$; since the matrix $C$ is constant and the Hamiltonian part $J\nabla H$ of the vector field $(J-D)\nabla H$ is divergence-free, we have 
\begin{equation}
	L^* = \frac{1}{2} CC^\top\colon\nabla^2 - (J-D)\nabla H\cdot\nabla - D\colon\nabla^2 H\,.  
\end{equation} 
Now if $T_L=T_R=T$ is constant, then $2D=\beta CC^\top$ and thus 
\begin{equation}
	L^*\exp(-\beta H) = 0\,.
\end{equation}
If $T_L\neq T_R$, then \cite[Theorem 1.1]{Carmona2007} implies that there exists a unique invariant measure $\mu_{\rm ne}\neq\mu_{\rm e}$, with a strictly positive density and the property
\begin{equation}\label{expConv}
	|\mu_{\rm ne}(f) - \mu_t(f)| \le M e^{-\lambda t}
\end{equation}
for some constants $M,\lambda>0$, where $\mu_t$ is the distribution of the process at time $t$, starting from an arbitrary initial distribution $\mu_0$, and $M=M(f,\mu_0)$. 

\subsection{The Gaussian case} A special situation is when the dynamics is linear, in which case the invariant measures $\mu_{\rm ne}$ and $\mu_{\rm e}$ are both Gaussian. We define the square matrix $A=(J-D)\nabla^2 H$ and assume that the Hamiltonian reads
 \begin{equation}\label{quadH}
 	H=\frac{1}{2}x^\top Sx\,,\quad 	S = \begin{pmatrix}
 		K & 0 \\ 0 & M^{-1}
 	\end{pmatrix},
 \end{equation}
for two symmetric and positive definite matrices $M,K\in\R^{Nd\times Nd}$. Then the Langevin dynamics (\ref{langPHS}) is given by 
\begin{equation}\label{ou}
	dx(t) = Ax(t)dt + CdW(t)\,,
\end{equation}
where the matrix $A$ is of the form 
\begin{equation}
		A = \begin{pmatrix}
		0 & M^{-1} \\ -K & -D
	\end{pmatrix}. 
\end{equation}
Since the matrix $A$ for $D=0$ is similar to a skew-symmetric matrix, with all eigenvalues on the imaginary axis, the perturbation by the friction matrix $D\ge 0$ moves all eigenvalues to the strict negative half plane (cf.~\cite{Overton1995}). 

As a consequence the matrix $A$ is stable. Moreover, since no left eigenvector $v$ of $A$ lies in the left kernel of $B$, i.e. $v^\top B\neq 0$, it follows from \cite[Theorem 4.15]{Antoulas2005} that the matrix pair $(A,C)$ is controllable. As a consequence, (\ref{langPHS}) has a unique Gaussian invariant measure $\mu_{\rm ne}=\cN(0,\Sigma)$ where the covariance matrix $\Sigma=\Sigma(T_L,T_R)$ is the unique positive definite solution to the Lyapunov equation
\begin{equation}\label{lyap}
	A\Sigma + \Sigma A^\top = -CC^\top\,.
\end{equation}
In the following we will call $\mu_{\rm ne}=\cN(0,\Sigma)$ the nonequilibrium steady state. Note that $\Sigma(T_L=T,T_R=T)=\beta^{-1}S^{-1}$ agrees with the equilibrium covariance.

\section{Free energy of the nonequilibrium steady state}\label{sec:freeEnergy}

We now consider equilibrium and nonequilibrium steady states, $\mu_{\rm e}$ and $\mu_{\rm ne}$ for a possibly nonlinear Hamiltonian $H$ where we assume that both have strictly positive densities $\rho_{\rm ne},\,\rho_{\rm e}>0$. 
Unlike in the equilibrium situation, there is no a priori notion of a free energy. For example, if $\mu_0$ and $\mu_1$ are two stationary distributions of an \emph{equilibrium} system at inverse temperature $\beta$, with densities $\rho_0\propto e^{-\beta H_0}$ and $\rho_1\propto e^{-\beta H_1}$ where $H_0$ and $H_1=H_0+U$ are two Hamiltonians, then  
\begin{equation}
    \frac{d\mu_1}{d\mu_0}(x) = e^{c(\beta) - \beta U(x)}\,,
\end{equation}
where $c(\beta)$ is a normalisation constant that guarantees that $(d\mu_1/d\mu_0)(x)$ is a probability density function in $x$. In particular, we can define the equilibrium free energy difference between the two quilibrium systems by the standard expression 
\begin{equation}
    \Delta F(\beta) := -\beta^{-1}\log\left(\frac{Z_1}{Z_0}\right)\,, \quad Z_i = \int e^{-\beta H_i(x)}\,dx\quad (i=0,1)\,,
\end{equation}
which shows that the normalisation constant $c(\beta)$ is given by $\beta \Delta F$.

\subsection{Free energy from exponential tilting}

For nonequilibrium systems, or the corresponding steady states the concept of relative free energy or free energy difference has no a priori meaning, yet the likelihood $d\mu_{\rm ne}/d\mu_{\rm e}=\rho_{\rm ne}/\rho_{\rm e}$ between the two steady states does. 

One possibility to measure the discrepancy between $\mu_{\rm ne}$ and $\mu_{\rm e}$ is by using a divergence measure such as the relative entropy (or: Kullback-Leibler divergence)
\begin{equation}
	\KL(\mu_{\rm ne},\mu_{\rm e}) = \int \log\left(\frac{d\mu_{\rm ne}}{d\mu_{\rm e}}\right)d\mu_{\rm ne}\,.
\end{equation}
Here, we assume that the relative entropy is well-defined, in that the log likelihood ratio $\log(d\mu_{\rm ne}/d\mu_{\rm e})=\log(\rho_{\rm ne}/\rho_{\rm e})$ is integrable \wrt the nonequilibrium distribution and there is no division by zero (i.e. $\mu_{\rm ne}$ is absolutely continuous \wrt $\mu_{\rm e}$); cf.~Definition \ref{defn:kl} in the appendix. 

We will argue that the relative entropy between the two distributions, implicitely defines a free energy difference between the corresponding steady states. Nevertheless, we do not want to challenge the widely accepted viewpoint that the free energy of a nonequilibrium system is related to the large deviations rate function in the large time or large particle limit (e.g.~\cite{Derrida2003,Fogedby2012}). In fact, our free energy may turn out to be a large deviations rate functions for the system under a suitable scaling of the variables (e.g. with the systems size or the diffusion coefficient).

\subsubsection*{A variational principle for the relative entropy}
 
To fix ideas, note that, by the Donsker-Varadhan variational principle, the relative entropy admits the following variational characterisation (see e.g.~\cite{DaiPra1996,DelleSite2017}):\begin{equation}\label{dv}
	\KL(\mu,\nu) = \sup\left\{\int \varphi \,d\mu  -\log\int e^{\varphi}\,d\nu \colon \varphi\in \cA(\nu)\right\} \,.
\end{equation}
where the set $\cA(\nu)$ contains all observables $\varphi$ that are bounded from below and for which $e^{\varphi}\in L^1(\nu)$. The supremum is attained for $\varphi$ satisfying 
\begin{equation}
	\frac{d\mu}{d\nu} = e^{\varphi - c(1)}\,, 
\end{equation}
where the normalisation constant $c(1)$ is essentially the cumulant generating function (CGF) of the observable $\varphi$ evaluated at $s=1$: 
\begin{equation}
	c(s) = \log\int e^{s\varphi}d\nu\,.
\end{equation}
Now letting $\mu=\mu_{\rm ne}$, $\nu=\mu_{\rm e}$ and redefining $\varphi:=-\beta G$,  $c(\beta):=-\log\bE_\nu\big[e^{-\beta G}\big]$, we see that we can interpret the nonequilibrium steady state $\mu_{\rm ne}$ as an exponentially tilted version of the equilibrium distribution $\mu_{\rm e}$. Specifically, 
\begin{equation}\label{potentialDef}
	\frac{d\mu_{\rm ne}}{d\mu_{\rm e}} = 	\frac{\rho_{\rm ne}}{\rho_{\rm e}} = e^{c(\beta)-\beta G}\,.
\end{equation} 
We call the function $G$ the \emph{potential} associated with the likelihood ratio between $\mu_{\rm ne}$ and $\mu_{\rm e}$.  In analogy with the equilibrium case, we can think of 
\begin{equation}
 -\beta^{-1}c(\beta) = -\beta^{-1}\log\bE_\nu\left[e^{-\beta G}\right] 
\end{equation}
as a free energy difference between the equilibrium and the nonequilibrium steady state. Note, however, that the interpretation of $\beta$ as a globally defined inverse temperature is misleading, for there is no such thing in the nonequilibrium system. The parameter $\beta$ should be rather be considered a scaling parameter that defines a characteristic energy scale, which also guarantees that $\beta G$ is a dimensionless physical quantity. 

As a consequence of the previous considerations, a two-sided Bogoliubov inequality holds:

\begin{thm}[Two-sided Bogoliubov inequality]\label{thm:bogo}
Assuming that a potential $G$ exists that satisfies (\ref{potentialDef}), with  $-\beta G\in\cA(\mu_{\rm e})$, then 
\begin{equation}\label{2bogo}
\bE_{\mu_{\rm ne}}[G] \le -\beta^{-1}\log\bE_{\mu_{\rm e}}\big[e^{-\beta G}\big] \le \bE_{\mu_{\rm e}}[G]\,.
\end{equation}
\end{thm}
\begin{proof}
Even though, the proof is very similar to the equilibrium case (e.g.~\cite{DelleSite2017,Reible2022}), we sketch it here for the reader's convenience. 
The upper bound is a straight consequence of Jensen's inequality and the convexity of the exponential function:  
\begin{equation}
   \bE_{\mu_{\rm e}}\big[e^{-\beta G}\big] \ge e^{-\beta\bE_{\mu_{\rm e}}[G]} \quad\Longrightarrow\quad -\beta^{-1}\log\bE_{\mu_{\rm e}}\big[e^{-\beta G}\big] \le \bE_{\mu_{\rm e}}[G]\,.
\end{equation}
The lower bound follows from the Donsker-Varadhan principle (\ref{dv}): setting  $\mu=\mu_{\rm ne}$, $\nu=\mu_{\rm e}$, and $\varphi:=-\beta G$, it follows that  
\begin{equation}
	\KL(\mu_{\rm ne},\mu_{\rm e}) =  -\beta\bE_{\mu_{\rm ne}}[G] - \log\bE_{\mu_{\rm e}}\big[e^{-\beta G}\big] 
\end{equation}
where we have used the assumption that the supremum is attained at $-\beta G\in\cA(\mu_{\rm e})$. Since the KL divergence is nonnegative (and finite), it follows that   
\begin{equation}\label{dv2}
	 \bE_{\mu_{\rm ne}}[G] \le \beta^{-1}\KL(\mu_{\rm ne},\mu_{\rm e}) + \bE_{\mu_{\rm ne}}[G] = -\beta^{-1}\log\bE_{\mu_{\rm e}}\big[e^{-\beta G}\big]\,,
\end{equation}
which proves the lower bound. 
\end{proof}

\begin{rem}
 In many cases, $\beta G$ will be bounded below, but not $-\beta G$. Assuming suitable growth conditions on $G$ for $|x|\to\infty$ and a letting $\mu_{\rm e}$ decay sufficiently fast at infinity, it is  possible to extend the Donsker-Varadhan principle to this case.
\end{rem}

Note that in contrast to the equilibrium situation (cf.~\cite{DelleSite2017,Reible2023}), when $G$ is basically a difference between two Hamiltonians, the function $G$ is only implicitely defined as a likelihood ratio between two probability density functions. 

Yet, estimating the upper and lower bounds is not hopeless: up to an additive constant, $G$ is equal to the difference between the log-density $\beta^{-1}\log\rho_{\rm ne}$ and the Hamiltonian $H$, where the constant appears on both sides, hence plays no role when we are just interested in the difference between the upper and lower bounds; moreover, the (unnormalised) log-density can be routinely estimated  during Monte Carlo simulations, for example, by using kernel density estimators (see e.g. \cite{Sullivan1988,Hazelton2016}).

\subsection{Transient behaviour and relaxation to (non)equilibrium}

Equation (\ref{dv2}) in the proof of the Bogoliubov inequality can be used to derive a modified bound. To this end, we replace $\mu_{\rm ne}$ in (\ref{dv2}) by $\mu_t$, the time-$t$ marginal of (\ref{langNoneq}), starting from the equilibrium distribution $\mu_0=\mu_{\rm e}$ at $t=0$. Moreover, we call 
\begin{equation}\label{langPHS0}
	dz(t) = (J-D)\nabla H(z(t)) + \bar{C} d\bar{W}(t)
\end{equation}
the equilibrium system with the fluctuation-dissipation property $2D=\beta \bar{C}\bar{C}^\top$ and stationary distribution $\mu_{\rm e}$, with density $\rho_{\rm e}\propto e^{-\beta H}$. We further denote by 
\begin{equation}\label{langPHS1}
	dx(t) = (J-D)\nabla H(x(t)) + CdW(t)
\end{equation}
the corresponding nonequilibrium system coupled to two heat baths. (Since we are only interested in the statistical properties of the dynamics, we assume that both the equilibrium and the nonquilibrium dynamics are driven by inependent Brownian motions $\bar{W}$ and $W$ in $\R^{(|L|+|R|)d}$.)  In this case, $C\neq \bar{C}$, i.e. $2D\neq \beta CC^\top$, and there is in general no closed-form expression for the stationary distribution $\mu_{\rm ne}$. 

If we denote by $\rho_t$ the density associated with the distribution $\mu_t$ of the nonequilibrium dynamics at time $t$, then setting for $\mu=\mu_t$, $\nu=\mu_{\rm e}$ in (\ref{dv}) and combining it with the lower bound in (\ref{2bogo}) implies the time-dependent two-sided bound
\begin{equation}\label{transbogo}
    	 \bE_{\mu_{\rm ne}}[G] \le  -\beta^{-1}\log\bE_{\mu_{\rm e}}\big[e^{-\beta G}\big] \le (2\beta)^{-1}\int_0^t \cL(\rho_s),ds + \bE_{\mu_{t}}[G]
\end{equation}
that holds for $t\ge t^*$ and sufficiently large $t^*>0$. Here the Lagrangian
\begin{equation}
     \cL(\rho_t)  =   \frac{1}{4}\int\left|\left(\bar{C}\bar{C}^\top - CC^\top\right)\nabla\log\rho_t(x)\right|^2_{(\bar{C}\bar{C}^{\top})^{\sharp}}\,\rho_t(x)\,dx
\end{equation}
is the so-called Donsker-Varadhan rate functional that bounds the relative entropy according to Theorem \ref{thm:fir} (FIR inequality) in the appendix where $|h|_Q=\sqrt{h^\top Qh}$ is a weighted seminorm \wrt a symmetric matrix $Q\ge 0$  (see appendix). 

Clearly, the upper bound in (\ref{transbogo}) becomes trivial in the long-time limit, since $\cL(\rho_t)$ is bounded away from zero, and so the integral diverges as $t\to\infty$. Nevertheless, it can be useful, because it does not require to sample \emph{both} $\mu_{\rm e}$ and $\mu_{\rm ne}$; it rather depends only on the distribution $\mu_t$ of the \emph{nonequilibrium} dynamics that approaches the unique stationary distribution $\mu_{\rm ne}$ in the limit $t\to\infty$. 

The practical difficulty consists in determining a good value for $t$ that is large enough, so that $\mu_t$ yields an approximation of the lower bound $\bE_{\mu_{\rm ne}}[G] \approx\bE_{\mu_t}[G]$ , but that is not so large that the upper bounds degrades.

\subsection{Mean-variance approximation: small $\beta$ asymptotics}

Yet another approximation to the free energy is the so-called mean-variance approximation that does provide an upper or lower bound, but that becomes exact when the reference (equilibrium) distribution is Gaussian and the effective potential $G$ is linear. By Taylor expanding the free energy about $\beta=0$, it follows that 
\begin{equation}
     -\beta^{-1}\log\bE_{\mu_{\rm e}}\big[e^{-\beta G}\big] =  \bE_{\mu_{\rm e}}[G] - \frac{\beta}{2}{\rm Var}_{\mu_{\rm e}}(G) + O(\beta^2)\,.
\end{equation}
Since the free energy is in general neither convex nor concave as a function of $\beta$, we cannot expect the mean-variance approximation to provide an upper or lower bound. Yet, in practice, and this is empirically demonstrated below, the mean-variance approximation turns out to be a fairly accurate approximation that resides between the upper and lower bounds.

Moreover, and more importantly, it yields a controlled approximation of the free energy that only requires samples from the \emph{equilibrium} distribution $\mu_{\rm e}$.

\section{Case study: the Gaussian chain}\label{sec:cases}

We will now study different variations on the linear harmonic chain to validate the aforementioned bounds as functions of the system size and the strength of the nonequilibrium perturbation. 

\subsection{One-dimensional chain of identical oscillators}

To begin with, we consider a one-dimensional chain of oscillators, with quadratic Hamiltonian $H$ given by (\ref{linearChain}). In this case, the stationary equilibrium and nonequilibrium distributions are the zero-mean Gaussians
\begin{equation}
    \mu_{\rm e}=\cN(0,\beta^{-1}S)^{-1}\,,\quad \mu_{\rm ne}=\cN(0,\Sigma)\,,
\end{equation}
with $\Sigma$ being the unique solution to the Lyapunov equation (\ref{lyap}). Here we use the matrix-vector notation from page \pageref{quadH}, where without loss of generality, we set $\beta=1$ and $T_L=1$. Then the two distributions have the densities
\begin{equation}
    \rho_{\rm e} = \frac{\exp\left(-\frac{1}{2}x^\top S x\right)}{Z_{\rm e}}\,,\quad \rho_{\rm ne} = \frac{\exp\left(-\frac{1}{2}x^\top\Sigma^{-1}x\right)}{Z_{\rm ne}}\,,
\end{equation}
with 
\begin{equation}
    Z_{\rm e} = \frac{(2\pi)^{3Nd}}{\det(S)}\,,\quad  Z_{\rm ne} = (2\pi)^{3Nd}\det(\Sigma)
\end{equation}
The log likelihood ratio thus is
\begin{equation}
    \begin{aligned}
-\log\frac{\rho_{\rm ne}}{\rho_{\rm e}} = \frac{1}{2}\log\det(S\Sigma) - \frac{1}{2}x^\top(S - \Sigma^{-1})x  \,. 
    \end{aligned}
\end{equation}
Using stability results for Lyapunov equations (e.g.~\cite[Theorem 2.1]{Hewer1988}), it can be shown that the difference between $S$ and the equilibrium covariance $\Sigma^{-1}$, measured in the Frobenius norm, is essentially proportional to the Frobenius norm of the difference $\bar{C}\bar{C}^\top-CC^\top$. To see this, note that 
\begin{equation}
    S - \Sigma^{-1} = \Sigma^{-1}(\Sigma-S^{-1})S\,,
\end{equation}
where the matrix $\Sigma-S^{-1}$ is the difference between the nonequilibrium and the equilibrium covariances. The claimed error bound then simply follows from the submultiplicativity of the Frobenius norm, noting that the matrix in the middle on the right hand side solves the Lyapunov equation
\begin{equation}
    A(\Sigma-S^{-1}) + (\Sigma-S^{-1})A^\top = - (CC^\top - \bar{C}\bar{C}^\top)\,,
\end{equation}
where $A=(J-D)\nabla^2 H$. 
Note that the symmetric matrix $G=CC^\top - \bar{C}\bar{C}^\top$ is positive or negative semidefinite, depending on whether $T_L<T_R$ or $T_L>T_R$, which implies that the difference is positive or negative semidefinite. As a consequence, the free energy difference can be positive or negative too. 

\begin{rem}
The free energy $-\log\bE_{\mu_{\rm e}}\big[e^{-G}\big]$ exists even if the quadratic form $G(x)=\frac{1}{2}x^\top(\Sigma^{-1}-S)x$ in the exponent is negative  definite, since $\mu_{\rm e}$ is (sub-)Gaussian, which implies that its tails are sufficiently light to guarantee that the CGF exists; for details on sub-Gaussian probability distributions, we refer to \cite[Sec. 2.5]{Vershynin2018}
\end{rem}

\subsubsection*{Free energy approximation and empirical quality measure}

Specifically, we first consider the Hamiltonian (\ref{linearChain}) with $N=100$ particles, pinned ends and unit mass matrix $M=I$. The temperature at the leftmost (first) particle is set to $T_L=1$ whereas the temperature $T_R$ at the rightmost ($N$-th) particle is varied between $T_R=0.5$ and $T_R=1.5$. Figure \ref{fig:linearChainBounds} shows the free energy and its upper and lower bounds. The upper panel of the figure also shows the arithmetic mean 
\begin{equation}
   \overline{G} = \frac{\bE_{\mu_{\rm e}}[G] + \bE_{\mu_{\rm ne}}[G]}{2}
\end{equation}
that, by definition, provides an approximation of the free energy that is on average closer than the upper and lower bounds. This is illustrated in the lower panel of the figure that shows the relative error of the arithmetic mean and the mean-variance approximation. The mean-variance approximation turns out to approximate the free energy rather well, and thus yields a computationally cheap alternative to the arithmetic mean when the nonequilibrium steady state $\mu_{\rm ne}$ is not available, since it relies completely on \emph{equilibrium} samples.

\begin{figure}
    \centering
    \includegraphics[width=0.55\textwidth]{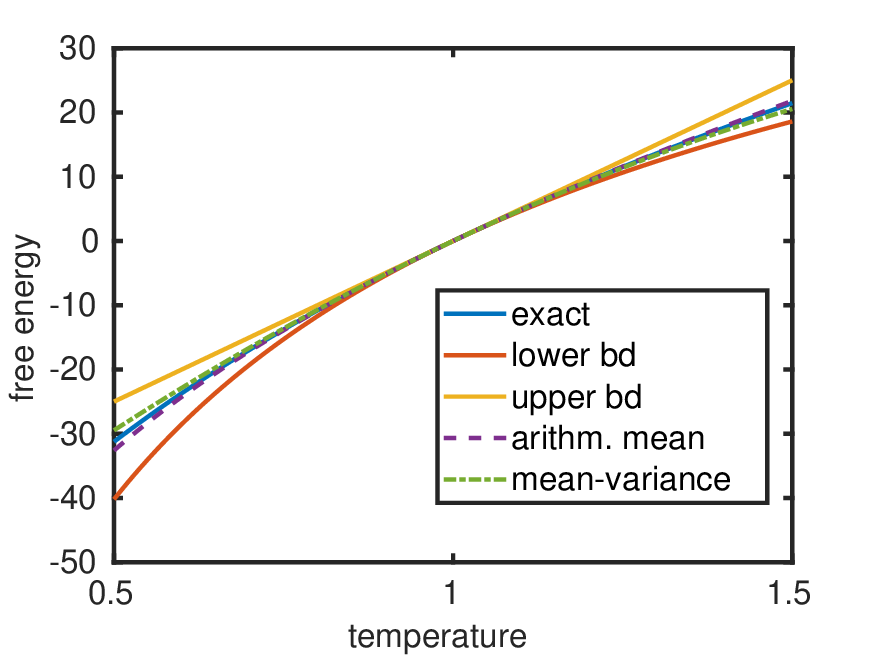}
      \includegraphics[width=0.55\textwidth]{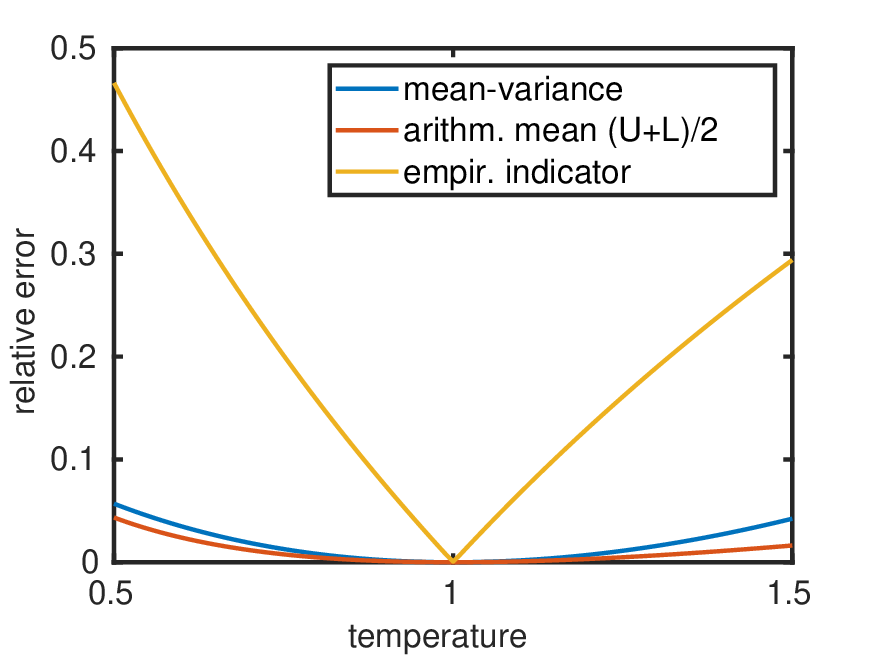}
    \caption{Free energy difference between equilibrium and nonequilibrium steady state and their upper and lower bounds: the upper panel shows the free energy and its approximations, the lower panel shows the exact relative error of the approximations together with the empirical quality measure $Q$ defined in (\ref{Q}).}
    \label{fig:linearChainBounds}
\end{figure}

Since the exact free energy is typically not known, a simple empirical quality measure for the tightness of the two bounds that does not rely on any other energy estimates is the difference between upper and lower bound divided by the mean: 
\begin{equation}\label{Q}
    Q = 2\frac{|\bE_{\mu_{\rm e}}[G] - \bE_{\mu_{\rm ne}}[G]|}{|\bE_{\mu_{\rm e}}[G] + \bE_{\mu_{\rm ne}}[G]|}
\end{equation}
This estimate of the relative error is relatively conservative as the lower panel of Figure \ref{fig:linearChainBounds} shows, yet it is an easy-to-use and computationally cheap criterion.

\subsubsection*{Robustness in the thermodynamic limit}

It turns out that for the homogeneous linear chain, the upper and lower bounds are remarkably robust for large system size. Basically, the upper and lower bounds share the scaling behaviour of the free energy difference as $N\to\infty$. 

To illustrate this observation, we computed the free energy per particle together with its upper and lower bounds per particle. Figure \ref{fig:linearChainLimit} shows the free energy difference, their bounds and the relative error associated with the arithmetic mean for increasing values of $N$. Not only does the free energy grow linearly in the system size (as one would expect), but also the upper and lower bounds do where the relative error saturates quickly and then stays essentially constant.

\begin{figure}
    \centering
    \includegraphics[width=0.55\textwidth]{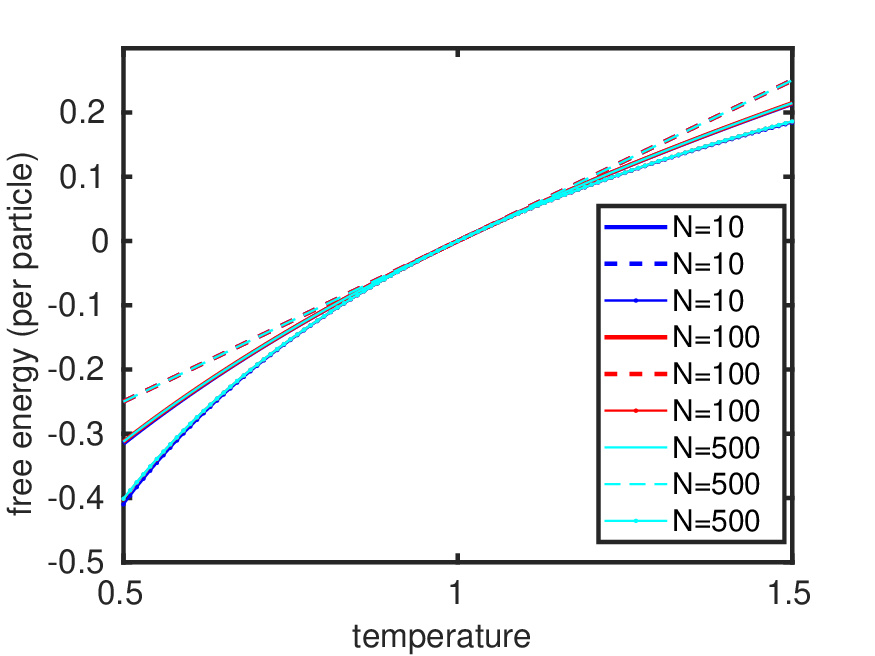}
    \includegraphics[width=0.55\textwidth]{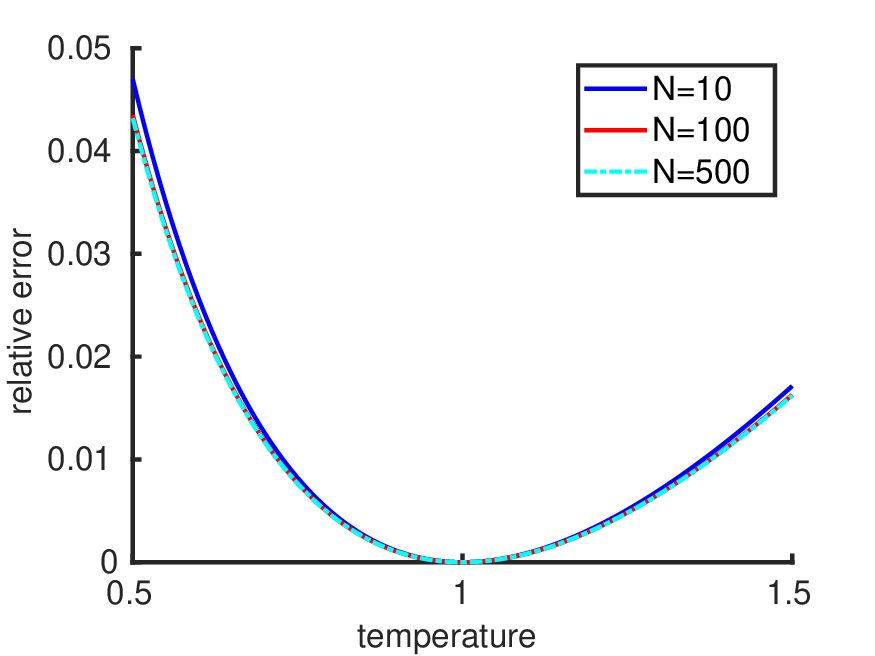}
    \caption{Thermodynamic limit of the free energy: The upper panel shows the free energy difference per particle for different numbers of particles, whereas the lower panel shows the corresponding relatve errors; for particle numbers beyond $N=10$, the free energy and its approximations quickly saturate.}
    \label{fig:linearChainLimit}
\end{figure}

Nonetheless the computation of both upper and lower bounds does become more challenging (as does the direct numerical computation of the free energy), since the principal eigenvalue of the matrix $A$ approaches the imaginary axis. This has consequences for the stable solution of the Lyapunov equation, but also for the approximation of the stationary distributions by Markov chain Monte Carlo (MCMC), since the speed of convergence to the nonequilibrium steady states is determined by the spectral gap of the matrix $A$; see \cite[Thm.~4.6]{Arnold2014}. 

\begin{rem}
The loss of stability as $N\to\infty$ can be understood by noting that the sum of all eigenvalues of the matrix $A$ is bounded: 
\begin{equation}
    \sum_{\lambda\in\sigma(A)} \lambda = {\rm tr}(A) = -{\rm tr}(D) = -2\gamma \,.
\end{equation}
Since the matrix $A$ has $2N$ (possibly multiple) eigenvalues with strictly negative real part, the one with the largest real part must decay to zero at least at a rate $O(\gamma/N)$. This implies that the system becomes unstable in the limit $N\to\infty$. In other words, the effect of the thermostat that is coupled only to the two boundary particles vanishes and the system becomes purely deterministic.  
\end{rem}

\subsubsection*{Transient behaviour I: FIR inequality}

Using the upper bound in (\ref{transbogo}), we can derive an easy-to-compute approximation of the upper bound that only involves $\mu_{\rm ne}$. To this end, we observe that if $\rho_t\approx\rho_{\rm ne}$, then the Lagrangian can be approximated by 
\begin{equation}
    \cL(\rho_t)\approx\cL(\rho_{\rm ne}) 
    = {\rm tr}\left((\bar{C}\bar{C}^\top - 2CC^\top +CC^\top(\bar{C}\bar{C}^\top)^\sharp CC^\top)\Sigma^{-1}\right)\,,
\end{equation}
where we taken advantage of the fact that $\cL(\rho_{\rm ne})$ involves a Gaussian integral that can explicitely solved. A na\"ive approximation of the upper bound  now is
\begin{equation}
    	 \bE_{\mu_{\rm ne}}[G] \le  -\beta^{-1}\log\bE_{\mu_{\rm e}}\big[e^{-\beta G}\big] \lesssim \frac{T}{2\beta}\cL(\rho_{\rm ne}) + \bE_{\mu_{\rm ne}}[G]\,,
\end{equation}
where the symbol $\lesssim$ means that the inequality holds for sufficiently large $T$. Figure \ref{fig:linearChainFIR} shows the approximation (\ref{transbogo}) under the stationary approximation of the Lagrangian for various values of $T$. It can be observed that, if $T$ is too small, then the approximtion underestimates the true free energy difference. However, for intermediate values of $T$ (here: $T=20$), there is a remarkable fit of the free energy profile, yet the appropriate choice of the time horizon $T$ is dimension-dependent and highly problem-specific. One reason is that the scaling of the lower bound in $N$ differs from the scaling of the Lagrangian term, in particular, the minimum time horizon that achieves an upper bound has been found to scale like $T=\cO(N)$, hence $T$ must be always adapted to the systems size. 

Overall, it is unclear whether the FIR-based approximation (\ref{transbogo}) is feasible in practice, unless a good estimate of the de-correlation time of the nonequilibrium dynamics is available.  We leave this question for future work.

\begin{figure}
    \centering
    \includegraphics[width=0.55\textwidth]{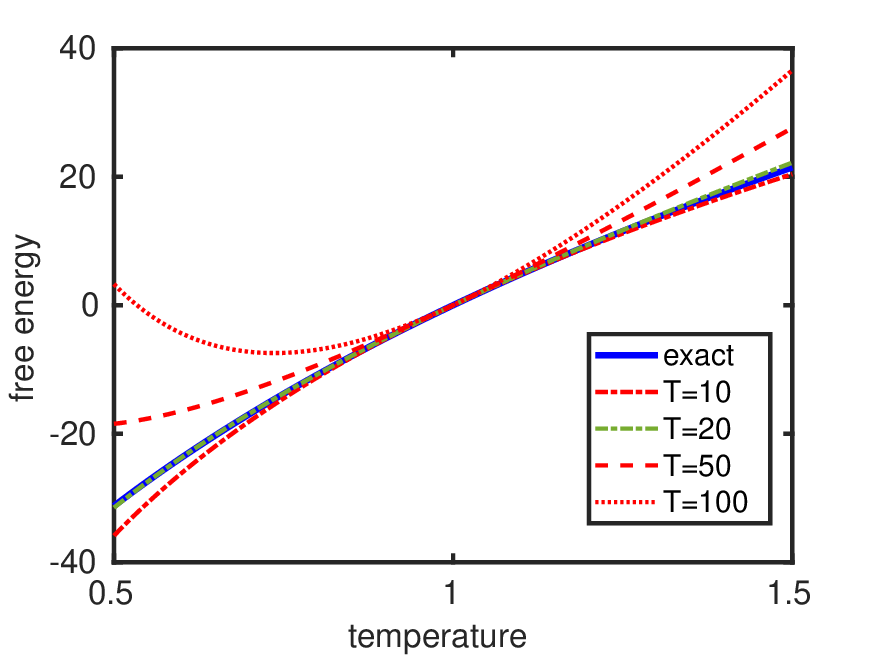}
    \caption{Approximation of the free energy based on the FIR inequality of Theorem \ref{thm:fir} and equation (\ref{transbogo}), using the stationary  approximation of the Lagrangian term $\cL$. Here the optimal terminal time $T$ is close to the de-correlation time of the system (dashed green curve); if $T$ is too small, then (\ref{transbogo}) does not hold.}
    \label{fig:linearChainFIR}
\end{figure}

\subsection{Application to heat conduction}

It has been proved \cite{Rieder1967} that the homogeneous oscillator chain does not satisfy Fourier's law, in that there is no temperature gradient inversely proportional to the chain length. The local temperature can be expressed in terms of the kinetic temperature, which for the $i$-th particle reads 
\begin{equation}
    \theta_i = \bE_{\mu_{\rm ne}}\!\left[p_i\frac{\partial H}{\partial p_i}\right] = {\rm Var}_{\mu_{\rm ne}}\!\left(\frac{p_i}{\sqrt{m_i}}\right)\,.    
\end{equation}
As the upper panel of Figure \ref{fig:linearChainTemp} shows for the case $T_L=1$ and $T_R=2$, with $N=2000$ particles, the temperature drops between the thermostatted boundary particles and their neighbours and then remains constant throughout the chain (dashed red curve). The constant kinetic temperature of all inner particles equals the average temperature (expect for a slight overshooting of the temperature close to the boundary).

The situation changes if the particle masses are not constant. The upper panel of Figure   \ref{fig:linearChainTemp} shows the kinetic temperature profile for a system of $N=2000$ particles when all  particles have different masses (solid blue curve). Even though there is still a temperature jump  between the boundary particles and their nearest neighbours, the inner particles display a roughly linear temperature profile in accordance with Fourier's law. The observed behaviour is relatively robust with regard to changes in the distribution of the particle masses (as long as the matrix $A$ remains stable, with the principal eigenvalue being well separated from the imaginary axis).   

The effect of the non-uniform mass distribution on the system can be quantified by computing the free energy between the two nonequilibrium steady states
\begin{equation}
    \mu_0:=\cN(0,\Sigma_0)\,,\quad \mu_1:=\cN(0,\Sigma_1)
\end{equation}
for the nonequilibrium system with unit mass matrix and $\Sigma_0=\Sigma_{\rm ne}$, and for the nonequilibrium system with nonuniform mass matrix. Then 
\begin{equation}
    -\log\frac{d\mu_1}{d\mu_0} = \frac{1}{2}x^\top\left(\Sigma_1^{-1} - \Sigma_0^{-1}\right)x + \frac{1}{2}\log\frac{\det\Sigma_1}{\det\Sigma_0}\,,
\end{equation}
where the first summand defines the nonequilibrium potential $G$, whereas the second summand is the free energy difference between the two states. The lower panel of Figure \ref{fig:linearChainTemp} shows the free energy difference, together with its upper and lower bounds and arithmetic mean and mean-variance approximations.

It is probably fair to say that heat conduction in oscillator lattices is still not fully understood -- even for low-dimensional models. We refer to \cite{Benenti2023,Dhar2008} and the references therein for two comprehensive reviews. However, we stress that it is not the aim of this paper to improve the state-of-the-art of the vast field of heat conduction models, we rather want to look at existing model from the perspective of free energy computation. These models are interesting as they are prime examples of nonequilibrium systems, for which the computation of free energy profiles or, likewise, cumulant generating functions or large deviations rate functions plays an important role \cite{Saito2011}. 

\begin{figure}
    \centering
    \includegraphics[width=0.55\textwidth]{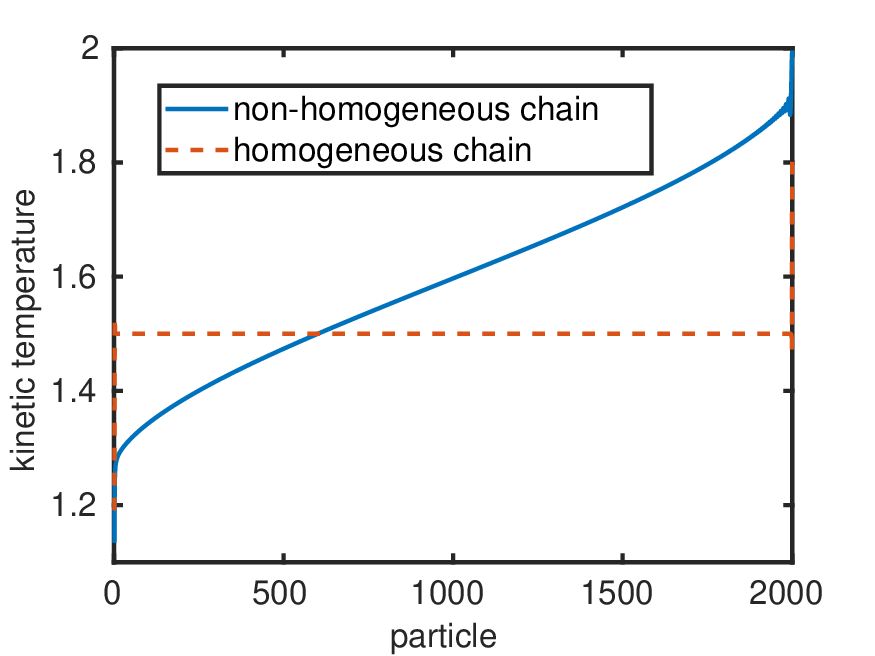}
        \includegraphics[width=0.55\textwidth]{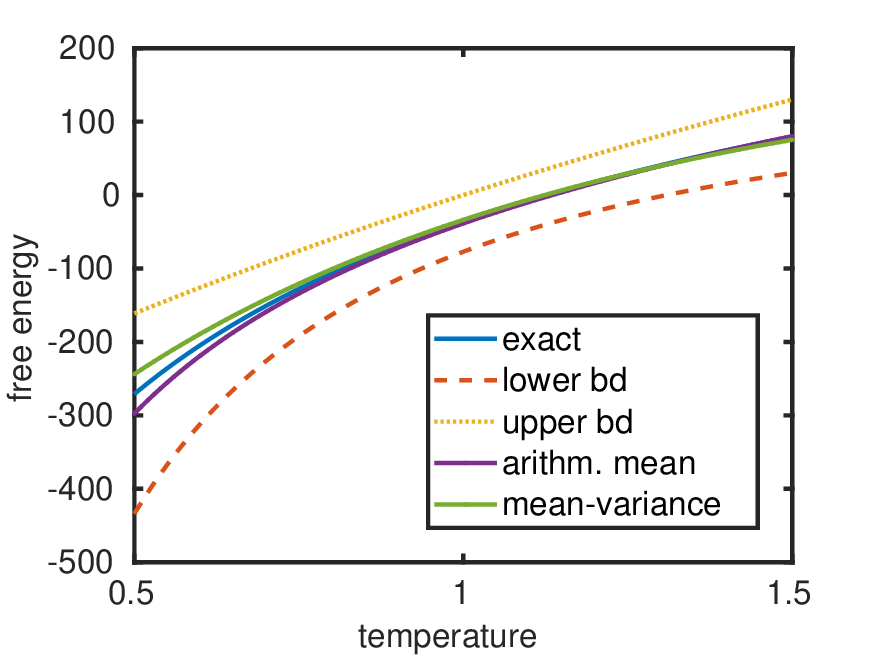}
    \caption{Kinetic temperature along the linear chain of particles (upper panel) and the free energy difference between the homogeneous and inhomogeneous chains (lower panel); the particles masses on the inhomogeneous system vary between $1/20$ and $2$ times the mass of the homogeneous chain with identocal particles. }
    \label{fig:linearChainTemp}
\end{figure}

\subsubsection*{Transient behaviour II: finite-time heat flux} We briefly mention yet another variant of the previous bounds to genuinely time-dependent phenomena; details are left for future work.  For the sake of clarity, we confine our discussion to the dynamics (\ref{langNoneq}) for $L=\{1\}$ and $d=1$.   
Suppose that we initialise the dynamics from an arbitrary initial distribution at time $t=0$. Then the total (negative) heat $Q$ transmitted in the time interval $[0,T]$ is given by \cite{Fogedby2012}
\begin{equation}
        \cQ := \gamma\int_0^T  |p_1(t)|^2 dt -  \sqrt{2\gamma T_L} \int_0^T p_1(t)\,dW_1(t)
\end{equation}
where we set again $T_L=1$. The quantity of interest to characterise the fluctuations of the heat flux is the cumulant generating function (CGF) of $Q$, i.e.
\begin{equation}\label{CGF}
    c(\alpha) = -\log\bE\!\left[e^{-\alpha \cQ}\right]\,,
\end{equation}
where the expectation is now taken over all realisations of the nonequilibrium dynamics (\ref{langNoneq}). The CGF does not yet have standard form due to the appearance of the stochastic (It\^o) integral in (\ref{CGF}). We can, however, transform it to our standard quadratic form by a Girsanov transformation (e.g.~\cite[Thm.~8.6.6]{oksendal2003stochastic}), which amounts to a change of the friction coefficient in the Langevin equation: Denoting by $\bP$ the measure on the space $C_0([0,T])$ of continuous trajectories with initial value $0$, such that $W$ is a standard Brownian motion, and defining 
\begin{equation}\label{newBM}
    \tilde{W}_1(t) = W_1(t) - \int_0^T u(t) dt\,,
\end{equation}
with $u=-\alpha\sqrt{2\gamma T_L}v_1=-\alpha \sqrt{2\gamma}v_1$, then $\tilde{W}_1$ is a standard Brownian motion with respect to a new probability measure $\bQ$, with 
\begin{equation}
    \frac{d\bQ}{d\bP} = \exp\left(- \frac{1}{2}\int_0^T |u(t)|^2\,dt + \int_0^T u(t) dW_1(t)\right)\,.
\end{equation}
Doing a change of measure from $\bP$ to $\bQ$, it follows that 
\begin{align*}
    \bE\!\left[e^{-\alpha Q}\right] & =  \bE_{\bQ}\!\left[e^{-\alpha\cQ} \frac{d\bP}{d\bQ}\right]\\
    & = \bE_{\bQ}\!\left[e^{-\alpha \mathcal{J}_\alpha(p_1)}\right]\,,
\end{align*}
with 
\begin{equation}
    \mathcal{J}_\alpha(v) = (1+\alpha)\gamma\int_0^T |v(t)|^2\,dt\,.
\end{equation}
Plugging (\ref{newBM}) into the Langevin equation (\ref{langNoneq}), we see that  under the new probability measure $\bQ$ the equation for $p_1$ turns into 
\begin{equation}
		dp_1 = (q_2 - 2q_1) dt - (1+2\alpha)\gamma p_1\,dt + \sqrt{2\gamma T_L}d\tilde{W}_1 
\end{equation}
Switching back to the original probability space with probability $\bP$ and expectation $\bE[\cdot]$, and using that the path distribution of $\tilde{W}_1$ under $\bQ$ is the same as the distribution of $W_1$ under $\bP$, we obtain the following representation of the CGF: 
\begin{equation}\label{CGF2}
    c(\alpha) = -\log\bE\!\left[\exp\left(-\alpha+\alpha^2)\gamma\int_0^T |p_1(t)|^2\,dt\right)\right],
\end{equation}
where the expectation is over all realisations of the nonequilibrium Langevin equation with modified friction for the leftmost particle: 
\begin{equation}\label{langNoneq2}
	\begin{aligned}
		dq_i & = \frac{p_i}{m_i}dt\,, \quad i=1\ldots N\\
		dp_i & = (q_{i-1} - 2q_i + q_{i+1})\,dt\,, \quad 1<i<N\\
		dp_1 & = (q_2 - 2q_1) dt - (1+2\alpha)\gamma p_1\,dt + \sqrt{2\gamma T_L}d W_1 \\
		dp_N & = (q_{N-1} - 2q_N)\,dt - \gamma p_N\,dt + \sqrt{2\gamma T_R}dW_N  \\
	\end{aligned}
\end{equation}
Our quantify of interest being a CGF, a similar argument as in the proof of Theorem \ref{thm:bogo} shows that it satisfies the following bound  
\begin{equation}
 \bE\!\left[\mathcal{J}_0(p_1)\right]   \le  -\frac{1}{\alpha}\log\bE\!\left[e^{-\alpha \mathcal{J}_\alpha(p_1)}\right]\le \bE\!\left[\mathcal{J}_\alpha(p_1)\right].
\end{equation}
We refrain from digging deeper here and leave the analysis of the finite-time case for future work. 

\begin{rem}
    For anharmonic chains, such as the Fermi-Pasta-Ulam chain, all bounds from the previous paragraphs can be used, and we expect to see similar trends in the approximation fidelity. The main difficulty consists in the estimation of the potential $G$ that requires the estimatation of log-densities. We will address the nonlinear case in  future work and refer to \cite{Sullivan1988,Hazelton2016} for the estimation of log-densities. 
\end{rem}


\section{Conclusions}\label{sec:conc}
The relevant methodological message of the paper is that 
the interpretation of the free energy as a cumulant generating function helps for a systematic derivation of (computationally simple) bounds and the systematic extension to nonequilibrium systems.
In this context, we have discussed the estimate of the free energy difference between equilibrium and nonequilibrium generated in a chain of oscillators, which describes the situation out of equilibrium produced by coupling the system to two different thermostats. The estimate is done by defining the upper and lower bound to the free energy difference employing the two-sided Bogoliubov inequality and various derived estimates, such as the mean-variance approximation. 

Numerical results for the heat conduction based on the free energy approach are shown for the different variants of the Gaussian chain. An interesting result is the difference between the homogeneous (particles with the same mass) and non homogeneous (particles with different masses) that reveals a rather different process for the heat conduction. Since it is widely recognized that the heat conduction in oscillator lattices is an ongoning field of research with many aspects still unknown, the perspective of free energy computation put forward in this paper may stimulate paths of research not yet considered in the field.

\section*{Acknowledgement} This article is devoted to Giovanni Ciccotti on the occasion of his 80th birthday. It was him who has introduced both authors of this article  to the (still) exciting field of free energy computations.

\appendix

\section{FIR inequality and relative entropy bounds}\label{sec:fir}

We want to compare the density evolutions of equilibrium and nonequilibrium dynamics. To this end, recall  (\ref{langPHS0}) that defines our reference dynamics, with constant heat bath and arbitrary Hamiltonian (i.e. linear or nonlinear), whereas (\ref{langPHS1}) defines the nonequilibrium dynamics, with two different heat baths. 

\subsection{FIR inequality}

To compare the time-$t$ marginals of the two equations (\ref{langPHS0}) and (\ref{langPHS1}), we first write the Fokker-Planck equation associated with (\ref{langPHS1}) as a continuity equation
\begin{equation}\label{langevinFPE}
\frac{\partial \rho}{\partial t} + \nabla\cdot\left(v\rho\right) = 0\,,
\end{equation}
with the time-dependent vector field
\begin{equation}
v_t(x) = (J - D)\nabla H(x) - \frac{1}{2}CC^\top \nabla\log\rho_t(x)\,.
\end{equation}
We use the shorthand notation $\rho_t=\rho(\cdot,t)$ for the smooth probability density of $x(t)=(q(t),p(t))\in\cX:=\R^{3Nd}\times\R^{3Nd}$ at time $t>0$. We further let 
\begin{equation}\label{langevinFPE2}
\frac{\partial \eta}{\partial t} + \nabla\cdot\left(\bar{v}\eta\right) = 0\,,
\end{equation}
be the Fokker-Planck equation associated with the reference dynamics (\ref{langPHS0}), with 
\begin{equation}
\bar{v}_t(x) = (J-D)\nabla H(x) - \frac{1}{2}\bar{C}\bar{C}^\top \nabla\log\eta_t(x)\,.
\end{equation}
To compare the two probability density functions $\rho_t$ and $\eta_t$, we will employ the relative entropy, also known as Kullback-Leibler (KL) divergence that provides an upper bound for their $L^1$ distance (see \cite{Csiszar1967}).

\begin{defn}[Relative entropy]\label{defn:kl}
	We define the relative entropy between two probability density functions $\rho,\eta\colon\cX\to[0,\infty)$, here with $\cX=\R^{3Nd}\times\R^{3Nd}$, as
	\begin{equation}\label{KL}
	\KL(\rho,\eta) = \int_{\cX}\log\frac{\rho(x)}{\eta(x)}\rho(x)\,dx
	\end{equation}
	provided that
	\begin{equation*}
	\int_{\{\eta(x)=0\}} \rho(x)\,dx = 0\,,
	\end{equation*}
	otherwise we set $\KL(\rho,\eta)=\infty$.
\end{defn} 
Assuming that $\rho_t$ is sufficiently decaying at infinity\footnote{Specifically, we assume that \begin{equation*}
	\lim_{|x|\to\infty} v_t(x)\rho_t(x) = 0\,,\quad 	\lim_{|x|\to\infty} \bar{v}_t(x)\rho_t(x) = 0\,,\quad
	\lim_{|x|\to\infty}  \bar{v}_t(x)\rho_t(x)\log\left(\frac{\rho_t(x)}{\eta_t(x)}\right) = 0\,.
	\end{equation*}} and that $\KL(\rho_t,\eta_t)<\infty$ for all $t>0$, it is straighforward to show that 
\begin{equation}\label{KLrate}
\frac{d}{dt} \KL(\rho_t,\eta_t) = \int_{\cX} \nabla\log\left(\frac{\rho_t}{\eta_t}\right)\cdot\left(v_t-\bar{v}_t\right)\rho_t \,dx\,,
\end{equation}
since (using a dot to denote the partial derivative \wrt $t$)
\begin{align*}
\frac{d}{dt} \KL(\rho_t,\eta_t) & = \int_{\cX} \left(\frac{\dot\rho_t}{\rho_t}-\frac{\dot\eta_t}{\eta_t}\right)\rho_t\,dx + \int_{\cX} \dot{\rho}_t\log\left(\frac{\rho_t}{\eta_t}\right)dx\\
 & = \int_{\cX} \left(-\nabla\cdot(v\rho_t) + \frac{\rho_t}{\eta_t}\nabla\cdot(\bar{v}\eta_t) - \log\left(\frac{\rho_t}{\eta_t}\right)\nabla\cdot(v\rho_t)\right)dx\\
  & = \int_{\cX} \left(\nabla\log\left(\frac{\rho_t}{\eta_t}\right)\cdot(v\rho)-\nabla\left(\frac{\rho_t}{\eta_t}\right)\cdot(\bar{v}\eta_t)\right)dx\\
  & = \int_{\cX} \nabla\log\left(\frac{\rho_t}{\eta_t}\right)\cdot\left(v-\bar{v}\right)\rho_t\,dx\,.
\end{align*}
Here we have substituted (\ref{langevinFPE}) and (\ref{langevinFPE2}) in step 2 and used integration by parts in step 3, assuming fast decay of the density $\rho_t$ which allowed us to ignore boundary terms. The last expression can be turned into an inequality connecting relative entropy, Fisher
information (FI) and a large deviations rate functional (R); for this reason it is termed FIR inequality (cf.~\cite{Roeckner2016,Pavon2006,Sharma2017}). 

Here we give only a formal derivation of the FIR inequality. To this end, we will need the following assumption:

\begin{ass}\label{ass:noise}
    The two noise coefficients have the same range and kernels, i.e.
    \begin{equation}
        {\rm ran}(C) = {\rm ran}(\bar{C})\,,\quad \ker(C) = \ker(\bar{C})
    \end{equation}
\end{ass}

\begin{thm}[FIR inequality, e.g.~\cite{Sharma2017}]\label{thm:fir}
	For all $\tau>0$ and under Assumption \ref{ass:noise}, 
	\begin{equation}\label{FIR}
	\frac{d}{dt} \KL(\rho_t,\eta_t)\le  \frac{\tau-1}{2}\mathcal{F}(\rho_{t},\eta_{t}) + \frac{1}{2\tau} \cL(\rho_{t})\,,
	\end{equation}
where 
\begin{equation}\label{FI}
\mathcal{F}(\rho_t,\eta_t) = \int_{\cX} \left|\nabla\log\left(\frac{\rho_t(x)}{\eta_t(x)}\right)\right|_{\bar{C}\bar{C}^\top}^2\rho_t(x)\,dx
\end{equation}
is the relative Fisher information, with $|z|^2_Q=z^\top Q z$ denoting the $Q$-weighted Euclidean vector seminorm with a semidefinite weight matrix $Q=\bar{C}\bar{C}^\top\ge 0$, and 
\begin{equation}\label{rate}
	\cL(\rho_t) = \frac{1}{4}\int_{\cX}\left|\left(\bar{C}\bar{C}^\top - CC^\top\right)\nabla\log\rho_t\right|_{(\bar{C}\bar{C}^{\top})^{\sharp}}^{2}\rho_t \,dx
\end{equation}
is called the Donsker-Varadhan large deviations rate functional. The expression $|z|^2_{Q^{\sharp}}=z^\top Q^\sharp z$ is understood similarly to $|z|_{Q}$ above, with the key difference that only the invertible block of $Q$ is inverted, and all other entries are set to zero (i.e. $Q^\sharp$ is the Moore-Penrose pseudoinverse of the semidefinite matrix $Q=\bar{C}\bar{C}^{\top}$).
\end{thm}

\begin{proof}
We start from the representation (\ref{KLrate}). Substituting $v$ and $\bar{v}$ we obtain by adding a zero: 	
\begin{align*}
\frac{d}{dt} \KL(\rho_t,\eta_t) = & \int_{\cX} \nabla\log\left(\frac{\rho_t}{\eta_t}\right)\cdot\left(\frac{1}{2}\bar{C}\bar{C}^\top \nabla\log\eta_t - \frac{1}{2}CC^\top \nabla\log\rho_t\right)\rho_t \,dx\\
 = &  -\frac{1}{2}\int_{\cX} \left|\nabla\log\left(\frac{\rho_t}{\eta_t}\right)\right|_{\bar{C}\bar{C}^\top}^2\rho_t\,dx\\ & +  \frac{1}{2}\int_{\cX} \nabla\log\left(\frac{\rho_t}{\eta_t}\right)\cdot\left(\left(\bar{C}\bar{C}^\top - CC^\top\right)\nabla\log\rho_t\right)\rho_t \,dx\\
 = & -\frac{1}{2}\mathcal{F}(\rho_t,\eta_t) + \frac{1}{2} \int_{\cX} \nabla\log\left(\frac{\rho_t}{\eta_t}\right)\cdot\left(\left(\bar{C}\bar{C}^\top - CC^\top\right)\nabla\log\rho_t\right)\rho_t \,dx\,.
\end{align*}
We set $Q=\bar{C}\bar{C}^\top - CC^\top$ and consider the second integral. Using Cauchy-Schwarz and Young's inequality, we find under Assumption   \ref{ass:noise}: 
\begin{align*}
I = & \frac{1}{2}\int_{\cX} \nabla\log\left(\frac{\rho_t}{\eta_t}\right)\cdot\left(Q\nabla\log\rho_t\right)\rho_t \,dx\\ 
 = & \frac{1}{2}\int_{\cX} \nabla\left(\frac{\rho_t}{\eta_t}\right)\cdot\left(Q\nabla\log\rho_t\right)\eta_t \,dx\\
= & \frac{1}{2}\int_{\cX} \nabla\left(\frac{\rho_t}{\eta_t}\right)\cdot\left(\bar{C}\bar{C}^\sharp Q\nabla\log\rho_t\right)\eta_t \,dx\\
= & \frac{1}{2}\int_{\cX} \bar{C}^\top\nabla\left(\frac{\rho_t}{\eta_t}\right)\frac{\eta_{t}}{\sqrt{\rho_{t}}}\cdot \left(\bar{C}^{\sharp}Q\nabla\log\rho_t\right)\sqrt{\rho_{t}} \,dx\\
\le & \sqrt{\int_{\cX} \left| \nabla\left(\frac{\rho_t}{\eta_t}\right)\right|^{2}_{\bar{C}\bar{C}^{\top}}\frac{\eta_{t}^{2}}{\rho_{t}}\,dx}\sqrt{\frac{1}{4}\int_{\cX}\left|Q\nabla\log\rho_t\right|_{(\bar{C}\bar{C}^{\top})^{\sharp}}^{2}\rho_t \,dx}\\
\le & \frac{\tau}{2}\int_{\cX} \left| \nabla_y\left(\frac{\rho_t}{\eta_t}\right)\right|^{2}_{\bar{C}\bar{C}^{\top}}\frac{\eta_{t}^{2}}{\rho_{t}}\,dx + \frac{1}{8\tau}\int_{\cX}\left|Q\nabla\log\rho_t\right|_{(\bar{C}\bar{C}^{\top})^{\sharp}}^{2}\rho_t \,dx
\end{align*}
for all $\tau>0$, where $(\bar{C}\bar{C}^{\top})^{\sharp}$ denotes the Moore-Penrose pseudoinverse of the block diagonal matrix $\bar{C}\bar{C}^{\top}$. Taking a closer look at the first integral, we see that  
\begin{equation*}
\int_{\cX} \left| \nabla\left(\frac{\rho_t}{\eta_t}\right)\right|^{2}_{\bar{C}\bar{C}^{\top}}\frac{\eta_{t}^{2}}{\rho_{t}}\,dx 
= \int_{\cX} \left| \nabla\log\left(\frac{\rho_t}{\eta_t}\right)\right|^{2}_{\bar{C}\bar{C}^{\top}}\rho_{t}\,dx 
= \mathcal{F}(\rho_{t},\eta_{t})\,.
\end{equation*}
As a consequence, 
\begin{equation}
\frac{d}{dt} \KL(\rho_t,\eta_t) \le \frac{\tau-1}{2}\mathcal{F}(\rho_{t},\eta_{t}) + \frac{1}{2\tau} \cL(\rho_{t})\,,
\end{equation}
which concludes the proof. 
\end{proof}

\subsection{A simplified FIR inequality for the total entropy production}

In (\ref{FIR}) we may tighten the upper bound by minimising over $\tau>0$. From a practical point of view, however, it is more convenient to set $\tau=1$, as a consequence of which the relative Fisher information term vanishes. Assuming that $\rho_0=\eta_0=\rho_{\rm e}$, and integrating the entropy rate from $t=0$ to $t=T$, we have
\begin{equation}\label{totalent}
\KL(\rho_T,\rho_{\rm e})\le 
\frac{1}{8}\int_{0}^{T} \left|\left(\bar{C}\bar{C}^\top - CC^\top\right)\nabla\log\rho_t\right|^2_{(\bar{C}\bar{C}^{\top})^{\sharp}}\,\rho_t\,dx dt\,. 
\end{equation}
Note that (\ref{expConv}) implies that the average entropy production rate can be bounded from above by  
\begin{equation}
    R:= \lim_{T\to\infty}\frac{1}{8T}\int_{0}^{T} \left|\left(\bar{C}\bar{C}^\top - CC^\top\right)\nabla\log\rho_t\right|^2_{(\bar{C}\bar{C}^{\top})^{\sharp}}\,\rho_t\,dx dt
\end{equation}
with
\begin{equation}
    R = \frac{1}{8}\bE_{\mu_{\rm ne}}\!\left[\left|\left(\bar{C}\bar{C}^\top - CC^\top\right)\nabla\log\rho_{\rm ne}\right|^2_{(\bar{C}\bar{C}^{\top})^{\sharp}}\right]
\end{equation}
Further approximations of the relative entropy production will be discussed for the specific case considered in Section \ref{sec:cases}.

\end{document}